\documentclass[twoside]{article} 

\usepackage{enumerate}
\usepackage{amsmath}    
\usepackage{amsthm}     
\usepackage{amscd}      
\usepackage{amssymb}    

\usepackage{eucal}      
\usepackage{latexsym}   
\usepackage{graphicx}   
\usepackage{verbatim}   
\usepackage{gimmsy2e}   
\usepackage[square,authoryear]{natbib}
\usepackage{epstopdf}
\usepackage{color}
\usepackage[colorlinks=true, pdfstartview=FitV, linkcolor=blue, citecolor=blue, urlcolor=blue]{hyperref}

\newtheorem{rmk}[thm]{Remark}
\numberwithin{thm}{subsection}

\makeatletter
\@addtoreset{equation}{subsection}
\def\thesubsection{\thesection\Alph{subsection}}
\def\thesection{\arabic{section}}

\def\theequation{\thesubsection.\arabic{equation}}

\makeatother

\def\intprod{\mathbin{\hbox to 6pt{%
                 \vrule height0.4pt width5pt depth0pt
                 \kern-.4pt
                 \vrule height6pt width0.4pt depth0pt\hss}}}

\newcommand{\ns}{\mspace{-1.5mu}}             
\newcommand{\ps}{\mspace{1.5mu}}              

\definecolor{gray}{rgb}{0.5, 0.5, 0.5}
\definecolor{gold}{rgb}{0.255, 0.215, 0}
\definecolor{cyan}{rgb}{0.0, 0.2, 0.2}
\definecolor{pink}{rgb}{0.2, 0.1, 0.2}

\begin{document}

\let\Box\blacksquare
\newpage
\thispagestyle{empty}

\title
{\huge \bf Covariantizing Classical Field Theories
\\[1.5ex]}
\author{%
{\bf Marco Castrill\'on L\'opez} \\
Departamento de Geometr\' ia y Topolog\' ia\\[-2pt]
Facultad de Ciencias Matem\' aticas \\[-2pt]
Universidad Complutense de Madrid \\[-2pt]
28040 Madrid, Spain
\\
\and
{\bf Mark J. Gotay} \\
Pacific Institute for the Mathematical Sciences\\[-2pt]
University of British Columbia\\[-2pt]
Vancouver, BC  V6T 1Z2 Canada \\
}
\date{\today}

\thispagestyle{empty}

\maketitle

\begin{abstract}{\footnotesize
We show how to enlarge the covariance group of any classical field theory in such a way that the resulting ``covariantized'' theory is `essentially equivalent' to the original. In particular, our technique will render any classical field theory generally covariant, that is, the covariantized theory will be spacetime diffeomorphism-covariant and free of absolute objects. 
Our results thus generalize the well-known parametrization technique of Dirac and Kucha\v r. Our constructions apply equally well to internal covariance groups, in which context they produce natural derivations of both the Utiyama minimal coupling and St\"uckelberg tricks. }
\end{abstract}

\newpage


\addtocontents{toc}{\protect\vspace{5ex}}



\section{Introduction}

The Principle of General Covariance has had a long (and somewhat checkered) history in physics (\cite{No1993}). In fact, this principle is often taken to mean slightly different things and it is unclear what the physical ramifications of various formulations are. \cite{Anderson1967} contains a careful exposition of the issues surrounding this idea.

Here, motivated by \cite{Anderson1967} and \cite{Kuchar1988}, we adopt the following precise definition: A Lagrangian field theory is {\bfi generally covariant} provided (\emph{i}) the Lagrangian density is equivariant with respect to the action of spacetime diffeomorphisms, and (\emph{ii}) 
all fields are variational. General covariance is of course a desirable property for a physical theory to possess, and it is quite useful when investigating classical field theories from both  structural and computational points of view (\cite{GoMa2010}, \cite{Leok2004}). 

Unfortunately, there are various ways in which general covariance can fail to hold. The most obvious is that the Lagrangian density $\mathcal L$ just does not have the required behavior under changes of coordinates; for instance, for the 1 + 1 Klein--Gordon equation
\begin{equation*}
\phi_{,xt} +  \frac{1}{2}m^2 \phi = 0
\end{equation*}
 with mass $m$ on $\mathbb R^2$, we have 
\begin{equation}
\label{KGL1}
\mathcal L = \left(\phi_{,t} \phi_{,x} - \frac{1}{2}m^2\phi^2 \ns  \right)\! dt \wedge dx
\end{equation}
which is visibly coordinate-dependent.

 A straightforward way to repair \eqref{KGL1} is to rewrite  it as
\begin{equation}
{\mathcal L} = \frac12 \left(g^{\mu\nu}\phi_{,\mu}\phi_{,\nu} - m^2 \phi^2\right) \ns \sqrt{-\det g} \, d^{\ps 2}\ns x
\label{KGL2}
\end{equation}
where $g$ is the Minkowski metric. Now we have $\textup{Diff}(\mathbb R^2)$-equivariance, but at the price of introducing the (nonvariational) field variable $g$.\footnote{\ Introducing a metric is not the only way of making \eqref{KGL1} geometric; we shall give another one later.}

This brings us to another way that a theory may fail to be generally covariant: it may contain an ``absolute''  or ``background'' object in the sense of \cite{Anderson1967} (see also \cite{Post2007}), in which case (\emph{ii}) cannot be satisfied. In \eqref{KGL2} $g$ is an absolute object (``prior geometry,'' in the terminology of \S17.6 of \cite{MiThWh1973}); as long as $g$ is immutable, the theory is still not generally covariant, despite the fact that the Klein--Gordon equation, which now takes the form $g^{\mu\nu}\phi_{;\mu\nu} + m^2 \phi = 0,$ is certainly an invariant statement.\footnote{\ Thus the Principle of General Covariance does not simply mean that ``the laws of physics must be expressible as geometric relationships between geometric objects,'' a point emphasized by \cite{MiThWh1973}, \S 12.5.}

However, if we insist that $g$ be dynamic, then varying it in \eqref{KGL2} yields 
\begin{equation}
\frac{\delta \mathcal L}{\delta g_{\mu\nu}} = 0 \quad \Longleftrightarrow \quad \mathfrak T^{\mu\nu} = 0,
\label{EMSEM}
\end{equation}
where $\mathfrak T$ is the stress-energy-momentum tensor density of the Klein--Gordon field. This 
forces $\phi = 0$, so that we have produced a system which is not equivalent to that with which we started. To undo \emph{this} problem, we must modify the Lagrangian once again, possibly by adding in a source term for the metric $g$ (e.g., regard $g$ as gravity). But in doing so we would likely still obtain an inequivalent (albeit consistent) system. Failing this we cannot allow $g$ to be variational.

So even in this basic example, we have two starting points based on: \eqref{KGL1}, which is not spacetime diffeomorphism covariant but in which all fields are variational, and \eqref{KGL2}, which is the other way around. Thus goals (\emph{i}) and (\emph{ii}) are more difficult to accomplish---simultaneously---than might be expected.\footnote{\  \cite{Kr1917} showed that it is always possible to make any theory spacetime diffeomorphism covariant by introducing auxiliary variables. As stressed by \citeauthor{Dirac1964} [\citeyear{Dirac1951},\citeyear{Dirac1964}] and \cite{Kuchar1988}, the catch is to do this while keeping all fields variational, without materially changing the theory.} And clearly we need to be more subtle than simply ``turning background fields on.''

There are other contexts in which similar issues arise. We may have a field theory which does not have a desired covariance property (not necessarily related to spacetime diffeomorphisms); can we modify it in such as way as to attain covariance while maintaining equivalence with the original system? A well-known  example of what we have in mind is provided by gauge theory. In this case we have a Lagrangian which is invariant under the action of a finite-dimensional Lie group, and we would like to enlarge the invariance group to the full gauge group. This can be accomplished by means of the Utiyama minimal coupling trick which, however, necessitates introducing  a (nonvariational)  connection. 

In this paper we show how to solve these and similar problems by systematically introducing new 
\emph{dynamic} fields into the formalism; basically, subject to a few mild assumptions,
we prove that  \emph{we can enlarge the covariance groups of field theories as desired while leaving solution spaces essentially unchanged.}\footnote {\, In a precise sense to be elucidated below.} In particular, our constructions provide a prescription for making  \emph{any} field theory generally covariant, as well as lead to natural derivations of both the St\"uckelberg Lagrangian and the Utiyama minimal coupling trick (conveniently understood) as special cases.

\section{Setup}

As usual in classical field theory we start with a configuration bundle $Y \stackrel{\pi_{XY}}{\longrightarrow}  X$ whose sections, denoted $\phi$, are the fields under consideration. The dimension of $X$ is taken to be $n + 1$, and we suppose that $X$ is oriented.\footnote{\ Often $X$ is a genuine spacetime, but in principle it need not have a Lorentzian structure (e.g., topological field theories).} Let
\[
\mathcal{L} : J^r Y \to \Lambda^{n+1} X
\]
be an $r^{\rm th}$-order Lagrangian density, where $J^r Y$ is the $r^{\rm th}$ jet bundle of $Y $ and $\Lambda^{n+1} X$ is the space of top forms on $X$. Loosely following the notation of 
\cite{CLGoMa2008}, we use $ \left( x^\mu, y^A, y^A{}_{\mu_1}, \ldots y^A{}_{\mu_1\ldots \mu_r} \right)$ as coordinates on $ J^1 Y$. In these coordinates  we write
$$\mathcal{L} = L\! \left( x^\mu, y^A, y^A{}_{ \mu_1}, \ldots y^A{}_{\mu_1\ldots \mu_r}  \right) \! d^{\ps n+1} \ns x.$$
Evaluated on the ($r^{\rm th}$) jet prolongation of a section $\phi$, the Lagrangian density becomes a function of $(x^\mu, \phi^A ,  \phi^A{}_{, \mu_1}, \ldots \phi^A{}_{, \mu_1\ldots \mu_r} )$, where
$$  \phi^A := y^A \circ \phi \quad {\rm and} \quad \phi^A{}_{,  \mu_1\ldots \mu_s} := y^A{}_{ \mu_1\ldots \mu_s} \circ j^s\ns \phi = \frac{\partial^{\ps s}\ns\phi^A}{\partial x^ {\mu_1\ldots \mu_s}}\ps $$
for $1 \leq s \leq r$; we shall abbreviate this when convenient and simply write $\mathcal{L} ( j^r\ns\phi )$. 

Suppose we have a group $\mathcal{G}$ which acts on $Y$ by
bundle automorphisms, so that we have a  monomorphism $\mathcal{G} \to
\operatorname{Aut}(Y)$.  (We will often blur the distinction between $\mathcal{G}$
and its  image in
$\operatorname{Aut}(Y)$.)  Being an automorphism of $Y \to X$, each $\sigma_Y \in \mathcal G$ covers a diffeomorphism of $X$ which we denote by $\sigma_X$.

We say that a Lagrangian field theory  is $\mathcal{G}$-{\bfi co\-var\-i\-ant} if the
Lagrangian density $\mathcal{L}$ is  equivariant with respect to the induced
actions of $\mathcal{G}$ on $J^rY$ by prolongation and $\Lambda^{n+1}X$ by push-forward:  
for all $\sigma_Y \in \mathcal{G}$ and $\gamma \in J^r Y$, 
\begin{equation}
\mathcal{L} \ns \left( j^r\ns \sigma_Y(\gamma) \right) = 
\sigma_{X*} \big[\mathcal{L}(\gamma)\big].
\label{equi} 
\end{equation}

A key fact is that a covariance group acts by {\bfi symmetries}: it maps solutions of the Euler--Lagrange equations to solutions. This is established in \S4.2 of \cite{Olver1986}.

\section{General Covariance}

General covariance is concerned with the action of the diffeomorphism group of the base $X$. Since ${\rm Diff}(X)$ is not necessarily a subgroup of ${\rm Aut}(Y)$, we must assume that there exists a group embedding ${\rm Diff}(X) \to {\rm Aut}(Y)$ covering the identity
denoted $\sigma \mapsto \sigma_Y$. (Such embeddings typically exist. For instance, for tensor theories, such an embedding is given by push-forward. See \cite{KoMiSl1993} and \cite{Ja2009} for more information.) 
 We suppose that the fields $\phi$ have differential index $k$,\footnote{ The ``differential index'' of a field is the highest order derivative that appears in its transformation law under diffeomorphisms of the base $X$, cf. \cite{GoMa1992}.} so that $\sigma_Y$ depends at most on the $k^{\rm th}$-order derivatives of $\sigma_X$. (Again, this is usually the case. Tensor fields in particular have index $\leq 1$.)

If \eqref{equi} holds for all $\sigma \in {\rm Diff}(X)$, then we say that the field theory is {\bfi diffeomorphism covariant}.

\subsection{The Variational Case}

We begin by tackling the case when a given classical field  theory is not diffeomorphism covariant at the outset, while assuming that all fields originally present are variational.
By combining the ideas of \citeauthor{Dirac1964} [\citeyear{Dirac1951},\citeyear{Dirac1964}], as further developed in \citeauthor{Kuchar1973} [\citeyear{Kuchar1973},\citeyear{Kuchar1988}] and \cite{IsKu1985}, with notions from elasticity theory (\cite{MaHu1983}), we ``covariantize'' the field theory under consideration by explicitly including diffeomorphisms of $X$ as genuine fields. This yields in particular an intrinsic reformulation of the idea of treating coordinate changes as fields (``parametrization''). 

To accomplish our goal, we: 

\indent 1.  Introduce new fields, the \emph{covariance fields}. These are just diffeomorphisms $\eta: X \to X$ reinterpreted as sections of the product bundle $X^2 \to X$. 

\indent 2. Replace $Y$ and $J^r Y$  by 
${\tilde Y} =  X^2 \times_X Y$ and $J^{k+r} \tilde Y$.
Coordinates on $J^{k+r} \tilde Y$ are\footnote{\ Henceforth we use a tilde to denote objects in a
covariantized theory from those in the original theory when
they  differ. We use bars to notationally distinguish  the factor $\bar X \approx X$ in the fiber of $X^2$ (and quantities defined thereupon) from the base $X$ (and quantities defined thereupon) when confusion may arise. We further use lowercase greek indices for coordinates on the base, and lowercase  latin indices for coordinates along  the fiber; when we refer to the \emph{same} coordinate on $X$ and $\bar X$ we use corresponding letters for indices, e.g. $x^\mu$ and $x^m$.}
$$(x^{\mu},x^a,\ldots x^a{}_{\mu_1\ldots\mu_{k+r}},y^A,\ldots,y^A{}_{\nu_1\ldots\nu_r}).$$

\indent 3. Replace the Lagrangian density $\mathcal L: J^rY \to \Lambda^{n+1}X$ by
$\tilde{\mathcal L}: J^{k+r}\tilde Y \to \Lambda^{n+1}X$ defined by
\begin{equation}
\label{newLag}
\tilde {\mathcal L}\big(j^{k+r}(\eta,\phi)(x)\big) = \eta^*\! \left[\mathcal L\big(j^r(\eta_Y \circ \phi\circ \eta^{-1})(\eta(x))\big)\right].
\end{equation}
In general $\tilde{\mathcal L}$ will depend on only the $r$-jet of $\phi$, but on the $(k+r)$-jet of $\eta$ as $\phi$ has differential index $k$ (whence $\eta_Y$ will depend upon $j^k\eta$). But it may happen that the order of $\mathcal {\tilde L}$ is $ < k+r$.

Suppose for instance $k=0$ and $r=1$. Then in charts we have
\begin{equation}
\tilde L\big(x^{\mu},x^a,x^a{}_\mu,y^A,y^A{}_{\nu}\big) = L\big(x^m, \eta^A, ( \eta^A{}_{B}\ps y^B{}_{\ns \rho} +  \eta^A{}_\rho)\ps x^{\ps \rho}{}_n\big)(\det J),
\label{newL}
\end{equation}
where we abbreviate $\eta^A = y^A\circ \eta_Y$ and the Jacobian of the `coordinate change' $X\to \bar X$ is $J = (x^a{}_\mu)$ with $J^{-1} = (x^\mu{}_a)$. Here we have used the chain rule to compute
\begin{eqnarray*}
y^A{}_{a}\big(& \! \! j^1 \ns\ns &  \ns\ns (\eta_Y  \circ  \phi\circ \eta^{-1})(\bar x)\big)  :=  \frac{\partial (\eta_Y \circ\phi \circ \eta^{-1})^A}{\partial x^a}(\bar x)  \nonumber \\[2ex] & = &
\left(\left[ \frac{\partial(y^A \circ \eta_Y)}{\partial y^B}\big(\phi(\eta^{-1}(\bar x))\big)\right] 
\left[\frac{\partial (y^B \circ \phi)}{\partial x^\mu}( \eta^{-1}(\bar x))\right]
 \right. \\[2ex]
& {} &  \qquad + \left. 
\left[\frac{\partial (y^A\circ \eta_Y)}{\partial x^\mu}\right]\big(\phi(\eta^{-1}(\bar x))\big)\right)
 \left[ \frac{\partial( x^\mu\circ \eta^{-1})}{\partial x^a}(\bar x)\right]
\nonumber
\end{eqnarray*}
or simply
\begin{equation}
y^A{}_{a} =  \big( \eta^A{}_{B}\ps y^B{}_{\mu} + \eta^A{}_\mu\big)x^\mu{}_a
\label{cr}
\end{equation}
for short. Similar formul\ae\ hold for other values of $k$ and $r$.
\begin{rmk} \rm \label{deriv coupling}
Note that because of \eqref{cr} the covariance fields derivatively couple in the Lagrangian density. 
 \hfill $\blacklozenge$  
\end{rmk}


\begin{rmk} \rm \label{elasticity}
This construction is reminiscent of the ``material'' or ``body'' representation in elasticity theory. In this context, $X$ plays the role of the \emph{body manifold} with the $x^\mu$ being \emph{body coordinates}, while $\bar X$ that of the \emph{space manifold} with the $x^a$ being \emph{spatial coordinates}. A diffeomorphism $\eta: X \to \bar X$ is a  \emph{deformation}, and the tangent map $T\eta$ is the \emph{deformation gradient}. For a deformation-field pair $(\eta, \phi): X \to \tilde Y$, the \emph{spatial field multivelocities} 

$$
y^A{}_a\big(j^1( \eta_Y \circ \phi \circ \eta^{-1})\big) :=  \frac{\partial (y^A \circ \eta_Y \circ \phi \circ \eta^{-1})}{\partial x^a}
$$
are related to the \emph{body field multivelocities} 
$$
y^A{}_\mu(j^1\phi) :=  \frac{\partial (y^A\circ \phi)}{\partial x^\mu}
$$
by the chain rule as in \eqref{cr}.
 \hfill $\blacklozenge$  
\end{rmk}
\begin{rmk} \rm \label{comp}
In the computational setting (see \cite{Leok2004}) the two copies of spacetime correspond  to a computational domain $X$ (where the mesh is uniform) and a physical domain $\bar X$ (where the mesh is concentrated about regions of high gradients), and the discrete variational principle allows one to obtain both the discrete diffeomorphism $\eta$ of spacetime (corresponding to the distribution of nodal points in spacetime), and the field values over these points. This is unlike moving mesh methods, where the solution of the field values is separate from the choice of the distribution of the nodal points.
 \hfill $\blacklozenge$  
\end{rmk}

Our first goal is to obtain a field theory which is diffeomorphism covariant. Our construction achieves this, as we now show.

Let $\sigma \in {\rm Diff}(X).$ On sections $\psi$ of $Y$ we define the action
\begin{equation}
\sigma \cdot \psi =  \sigma_{Y} \circ \psi \circ \sigma^{-1}.
\label{action on psi}
\end{equation}
We extend this to an imbedding of $\textup{Diff}(X) \to \textup{Aut}(\tilde Y)$ by requiring that $\textup{Diff}(X)$ act trivially on $\bar X$; this is consistent with viewing $\eta$ as a point mapping of $X$ to $\bar X$. Thus we have
\begin{equation}
\sigma \cdot \eta = \eta \circ \sigma^{-1}.
\label{action on eta}
\end{equation}
\begin{thm}
\label{cov}
The modified field theory on $J^{k+r}{\tilde Y}$ with Lagrangian density \eqref{newLag} is $\textup{Diff}(X)$-covariant. 
\end{thm}
\begin{proof}
Suppose $\sigma \in \textup{Diff}(X)$. Then by construction and \eqref{action on eta}, \eqref{newLag} yields
\begin{eqnarray*}
\tilde{\mathcal L}\ns\left( j^{k+r}\sigma_{\tilde Y} (j^{k+r}(\eta,\phi)(x))\right) 
& = & \tilde{\mathcal L} \big( j^r(\sigma \cdot \eta,\sigma \cdot\phi)(\sigma(x))\big) \\[1 ex]
& = &(\sigma \cdot \eta)^* \left[\ps \mathcal L \ns \left(j^r((\sigma \cdot \eta)_Y \circ(\sigma \cdot \phi) \circ (\sigma \cdot \eta)^{-1})(\eta(x))\right) \right] \\[1 ex]
& = & (\eta \circ \sigma^{-1})^* \left[\ps \mathcal L\ns \left(j^r(\eta_Y \circ \phi\circ \eta^{-1})(\eta(x))\right) \right] \\ [1 ex]
& = &\sigma_*\ns \left[\ps \tilde{\mathcal L}\ns \left(j^{k+r}(\eta,\phi)(x)\right) \right] 
\end{eqnarray*}
as was to be shown.
\end{proof}

It follows that $ {\rm Diff}(X)$ sends solutions of the variational problem defined by $\mathcal{\tilde{L}}$ to solutions. Of course, the same need not be true for $\mathcal L$. 

Our second  goal is to ensure that  the modified field theory is ``essentially equivalent'' to the original one. This is also the case, as we now explain.

We first consider the Euler--Lagrange equations of the new Lagrangian density $\tilde{\mathcal L}$ corresponding to the fields $\phi$. For this, we compute the variation of the action integral for compactly supported vertical variations $(\eta,\phi _\epsilon)$ with $\phi _0 = \phi$. In this case \eqref{newLag} gives
\begin{eqnarray*}
\left. \frac{d}{d\epsilon}\right|_{\epsilon =0} \int _X \tilde{\mathcal{L}}\ns \left(j^{k+r} (\eta , \phi _\epsilon)\right) & = & 
\left. \frac{d}{d\epsilon}\right|_{\epsilon =0} \int _X \eta ^* \mathcal{L}\ns\left(j^r (\eta _Y \circ \phi _\epsilon \circ \eta ^{-1})\right)\\[12pt]
& = &  \left. \frac{d}{d\epsilon}\right|_{\epsilon =0} \int _X \mathcal{L}\ns\left(j^r (\eta _Y \circ \phi _\epsilon \circ \eta ^{-1})\right)\nonumber
\end{eqnarray*}
by the change of variables formula. Since $\eta_Y$ is fiber-preserving it is clear that $\eta _Y\circ \phi _\epsilon \circ \eta ^{-1}$ is a compactly supported vertical variation of $\eta _Y\circ \phi \circ \eta ^{-1}$. In fact, as $\eta _Y$ is an automorphism of $Y$, any compactly supported vertical variation of $\eta _Y\circ \phi \circ \eta ^{-1}$ can be obtained in this way. Hence, the derivative above vanishes if and only if $\eta _Y\circ \phi \circ \eta ^{-1}$ is a critical section of $\mathcal{L}$. 

Now suppose $\phi$ is fixed and consider compactly supported variations $\eta _\epsilon$ of the covariance field $\eta$.\footnote{\ We must vary $\eta$ within Diff$(X)$.} Then, reprising the previous calculation,
\begin{eqnarray*}
\left. \frac{d}{d\epsilon}\right|_{\epsilon =0} \int _X \tilde{\mathcal{L}}\ns\left(j^{k+r} (\eta _\epsilon , \phi )\right) 
& = & \left.\frac{d}{d\epsilon}\right|_{\epsilon =0} \int _X \mathcal{L}\ns\left(j^r ((\eta _\epsilon)_Y  \circ \phi \circ \eta _\epsilon{}^{-1}) \right).\nonumber
\end{eqnarray*}
(Note that $(\eta _\epsilon)_Y  \circ \phi \circ \eta _\epsilon{}^{-1}$ can be thought as a ``horizontal variation'' of the section $\eta _Y  \circ \phi \circ \eta ^{-1}$, cf. \cite{MaPaSh1998}.) To compute this integral,  let $\Theta_{\mathcal L}$ be a Lepagean equivalent of $\mathcal{L}$ on $J^{2r-1}Y$ (see \cite{Go1991}). As the Lagrangian density is the 0-contact part of $\Theta_{\mathcal L}$,
\begin{eqnarray*}
 \mathcal{L}\ns\left(j^r ((\eta _\epsilon)_Y  \circ \phi \circ \eta _\epsilon{}^{-1}) \right)
& = & \left(j^{2r-1} ((\eta _\epsilon)_Y  \circ \phi \circ \eta _\epsilon{}^{-1} )\right)^*\Theta_{ \mathcal L} \\[6pt]
& = & \left(\eta _\epsilon{}^{-1}\right)^*(j^{2r-1}\phi)^*\left(j^{2r-1} (\eta _\epsilon)_Y\right)^*\Theta_{ \mathcal L}. 
\end{eqnarray*}
Applying the change of variables formula, the derivative of the action now reads
\begin{equation*}
\left. \frac{d}{d\epsilon}\right|_{\epsilon =0} \int _X (j^{2r-1}\phi)^*\left(j^{2r-1} (\eta _\epsilon)_Y\right)^*\Theta_{ \mathcal L} 
  =   \int _X (j^{2r-1}\phi)^*(j^{2r-1} \eta_Y)^* \pounds_{j^{2r-1}W}\Theta_{ \mathcal L} 
\end{equation*}
where $\pounds$ denotes the Lie derivative and $W$ is the vector field on $Y$  defined by the flow $(\eta _\epsilon)_Y$. Undoing the change of variables and using Cartan's magic formula, this becomes
\begin{equation*}
\int_X \left(j^{2r-1} (\eta _Y \circ \phi \circ \eta ^{-1})\right)^* i_{j^{2r-1} W}d\Theta _{\mathcal{L}} +\int _{X} \left(j^{2r-1} (\eta _Y \circ \phi \circ \eta ^{-1})\right)^*d\ps i_{j^{2r-1}W}\Theta _{\mathcal{L}}.
\end{equation*}
By Stokes' theorem the second term vanishes for compactly supported $W$. But if $\eta _Y \circ \phi \circ \eta^{-1}$ is a critical section of $\mathcal{L}$ (so that the Euler--Lagrange equations of $\tilde{\mathcal{L}}$ with respect to\! $\phi$ are satisfied) then the first term in the above expression vanishes (see  \cite{KrKrSa2010}).  We have thus proven:
\begin{thm}
\label{mainresult}
Let $(\eta , \phi)$ be a section of $\tilde{Y}$. Then
\begin{itemize}
\item[{\rm (\emph{i})}] $(\eta , \phi)$ satisfies the  Euler--Lagrange equations of $\tilde{\mathcal{L}}$ with respect to $\phi$ if and only if $\eta \cdot \phi $ satisfies the  Euler--Lagrange equations of $\mathcal{L}$, that is,
\[
\frac{\delta \tilde {\mathcal L}}{\delta \phi}(\eta , \phi)=0 \iff \frac{\delta {\mathcal L}}{\delta \phi}(\eta \cdot \phi)=0
\]
\item[{\rm (\emph{ii})}] If $(\eta , \phi)$ satisfies the  Euler--Lagrange equations of $\tilde{\mathcal{L}}$ with respect to $\phi$ then the  Euler--Lagrange equations of $\tilde{\mathcal{L}}$ with respect to the covariance fields $\eta$ are vacuously satisfied for $(\eta , \phi)$, that is,
\[
\frac{\delta \tilde {\mathcal L}}{\delta \phi}(\eta , \phi)=0 \Longrightarrow \frac{\delta {\tilde {\mathcal L}}}{\delta \eta}(\eta , \phi)=0.
\]
\end{itemize}
\end{thm}

 If $\mathcal{S}$ and $\mathcal{\tilde S}$ denote the sets of solutions of the Euler--Lagrange equations of
$\mathcal{L}$ and ${\mathcal{\tilde L}}$ respectively, Theorem \ref{mainresult} can be summarized by saying that the map $(\eta,\phi) \mapsto (\eta,\eta \cdot \phi)$ establishes a one-to-one correspondence
\[
\mathcal{\tilde S} \to {\rm Diff}(X) \times \mathcal{S}.
\]

Furthermore, since $\tilde{\mathcal L}$ is  ${\rm Diff}(X)$-covariant, ${\rm Diff}(X)$ acts on $\mathcal{\tilde S}$. We may then use the surjection $(\eta,\phi) \mapsto \eta \cdot \phi$ to identify  $\mathcal{S}$ with the orbit space $\mathcal{\tilde S}/ {\rm Diff}(X)$. Thus while the covariantized theory will have many more solutions than the original system, corresponding to the (gauge) freedom in the choice of covariance fields,  their solutions spaces stand in the simple relation
\[
\mathcal{S} \approx \mathcal{\tilde S}/{\rm Diff}(X).
\]
Because of this we say that the covariantization process leaves solution spaces ``essentially unchanged.'' As the covariance fields have no dynamics, we may additionally say that the original and covariantized theories are ``essentially equivalent.''

\begin{rmk} \rm
Theorem \ref{mainresult} provides, to our knowledge, the first instance in which the essential equivalence of the original and covariantized systems has been explicitly stated and proved.
\hfill $\blacklozenge$  
\end{rmk}

\begin{rmk} \rm
The situation above amounts to a reduction process when viewed in reverse. More precisely, if 
we begin with a ${\rm Diff}(X)$-covariant Lagrangian density $\mathcal{\tilde{L}}$, reduction by the ${\rm Diff}(X)$-action would yield the Lagrangian density $\mathcal{L}$. Indeed, the behavior of the set of solutions
 $\mathcal{S}$ and $\tilde{\mathcal{S}}$ given just above is simply the
relation between the set of solutions of a ${\rm Diff}(X)$-covariant Lagrangian and that of its reduced counterpart.
 \hfill $\blacklozenge$  
\end{rmk}

It is worthwhile explicitly writing out the Euler--Lagrange equations for $\tilde {\mathcal L}$ in the special case when $k=0$ and $r=1$. These equations for the fields $\phi$ unsurprisingly turn out to be reparametrizations of those for $\mathcal L$. (We give an example in \S\ref{NM}.) The Euler--Lagrange equations for the covariance fields are more interesting; we may express them as follows.

First we compute 
\begin{eqnarray}
\frac{\partial \tilde L}{\partial x^a{}_\mu} & = & L \frac{\partial (\det J)}{\partial x^a{}_\mu}  +\frac{\partial L}{\partial x^a{}_\mu}  (\det J)\nonumber \\[6pt]
 \label{second} 
  & = &  L \frac{\partial (x^c{}_{\tilde \nu} C_c{}^{\tilde \nu})}{\partial x^a{}_\mu} + 
  \frac{\partial L}{\partial y^A{}_c}\frac{\partial y^A{}_c}{\partial x^a{}_\mu}  (\det J). 
 \end{eqnarray}
In obtaining the first term here we have employed a cofactor expansion of $\det J$ (in which $C_c{}^{\nu}$ is the cofactor of $x^c{}_\nu$ and hatted indices are not summed). Now
$$\frac{\partial (x^c{}_{\tilde \nu} C_c{}^{\tilde \nu})}{\partial x^a{}_\mu} = C_a{}^{\tilde \nu} \delta^\mu{}_{\tilde \nu} =  C_a{}^{\tilde \nu} (x^c{}_{\tilde\nu} x^\mu{}_c) =  x^\mu{}_a (\det J).$$
In the second term we use \eqref{cr} to evaluate
$$\frac{\partial y^A{}_c}{\partial x^a{}_\mu} = (\eta^A{}_B y^B{}_\nu + \eta^A{}_\mu) \frac{\partial x^\nu{}_c}{\partial x^a{}_\mu} = (\eta^A{}_B y^B{}_\nu + \eta^A{}_\mu) \left(- x^\nu{}_a x^\mu{}_c\right) = - y^A{}_a x^\mu{}_c.
$$
Substituting these results into \eqref{second}, we have
\begin{eqnarray}
\frac{\partial \tilde L}{\partial x^a{}_\mu} & = & \left(L\ps \delta^c{}_a -  \frac{\partial L}{\partial y^A{}_c}y^A{}_a\right)\ns x^\mu{}_c (\det J).\nonumber 
\end{eqnarray}

The Euler--Lagrange equations for the covariance fields may thus be written
\begin{eqnarray}
0 & = & \frac{\partial \tilde L}{\partial x^a} - D_\mu \! \left(\frac{\partial \tilde L}{\partial x^a{}_\mu}\right)
\label{eleta} \\[6pt]
& = & \frac{\partial L(\det J)}{\partial x^a}- x^b{}_\mu D_b\!\left( \left[  L\ps \delta^c{}_a -  \frac{\partial L}{\partial y^A{}_c}y^A{}_a  \right]     \ns    x^\mu{}_c (\det J)    \right) \label{2nd} \\[6pt]
& = & \left( \frac{\partial L}{\partial x^a} +  D_b\! \left[  \frac{\partial L}{\partial y^A{}_b}y^A{}_a - L\ps \delta^b{}_a  \right]     \right)  (\det J)
\label{condition} \\[6pt]
& = &  \left( \frac{\partial L}{\partial x^a} +  D_b\! \left( \frac{\partial L}{\partial y^A{}_b}\right)y^A{}_a + \frac{\partial L}{\partial y^A{}_b}y^A{}_{ab} - D_a L  \right)  (\det J) \nonumber
\end{eqnarray}
In going from \eqref{2nd} to \eqref{condition}, we have used  the ``Piola identity'': $D_\mu (x^\mu{}_c \det J) = x^b{}_\mu D_b(x^\mu{}_c \det J) \equiv 0$. We recast  the first factor in this last expression in the form
$$- \frac{\delta L}{\delta y^A}y^A{}_a +\left( \frac{\partial L}{\partial x^a} +   \frac{\partial L}{\partial y^A}y^A{}_a + \frac{\partial L}{\partial y^A{}_b}y^A{}_{ab} - D_a L\right) $$
The term in parentheses vanishes identically by definition of the total derivative, whence we again obtain   (\emph{ii}) of Theorem \ref{mainresult}.

However, referring to \eqref{condition} we see that the Euler--Lagrange equations with respect to $\eta$ also  take the form
\begin{equation}
\frac{\partial L}{\partial x^a} -  D_b\mathfrak t^b{}_a = 0
\label{t}
\end{equation}
where 
\begin{equation*}
\mathfrak t^c{}_a  = L\ps \delta^c{}_a -  \frac{\partial L}{\partial y^A{}_c}y^A{}_a\end{equation*}
is the canonical stress-energy-momentum (``SEM'') tensor density of $\mathcal L$.
Should $L$ be independent of the coordinates $x^a$ these reduce to conservation of energy-momentum in the spatial representation. 

\begin{rmk}
{\rm By analogy with elasticity theory (see p. 7, \cite{MaHu1983}), we may define the two-point canonical Piola-Kirchhoff tensor density ${\mathfrak p}$ to be the Piola transform of $\mathfrak t$:}
$${\mathfrak p}^\mu{}_a = \mathfrak t^c{}_a x^\mu{}_c (\det J). $$
{\rm Then from \eqref{eleta} we may equally well express \eqref{t} as}
\begin{equation*}
\frac{\partial \tilde L}{\partial x^a} -  D_\mu{\mathfrak p}^\mu{}_a = 0.
\end{equation*}
{\rm When the covariance fields are ignorable, these yield energy-momentum conservation in the material picture.}
 \hfill $\blacklozenge$  
\end{rmk}

Thus we have achieved our objective in the case when all fields are variational: we have taken a non-diffeomorphism covariant system and produced an essentially equivalent theory which \emph{is} ${\rm Diff}(X)$-covariant. We now study some examples.

\subsubsection{Example: Nonrelativistic Mechanics} \label{NM}

Issues of general covariance appear even in the most elementary systems, such as mechanics.
For instance, consider a particle moving in a Riemannian configuration space $(Q,g)$, so that $Y = \mathbb R \times Q$. Supposing that the action of $\textup{Diff}(\mathbb R)$ on $\mathbb R \times Q$ is trivial on the second factor (as is typically the case in mechanics), the usual Lagrangian density
\begin{equation}
{\mathcal L} = \left(\frac12m\ps  g_{AB} \ps \dot q^A \dot q^B  - V({q})\! \right)\!  dt
\label{partld}
\end{equation}
is not time reparametrization covariant.

To repair this we introduce a covariance field $T:\mathbb R \to \mathbb R$.\footnote{\ The introduction of the covariance field $T$ is a venerable technique in mechanics, see \S V.6 in \cite{Lanczos1970}. 
However, our point of view here is ``active,'' whereas the classical approach is ``passive.''}
As the time reparametrization group does not act on the configuration space,   \eqref{newL} simplifies considerably: 
\begin{equation*}
\tilde L(t,T,q^A,\dot T, \dot q^A) = L(T,q^A,\dot q^A \dot T^{-1})\dot T
\end{equation*}
where $\dot T := dT/dt$, and so \eqref{partld} is replaced by
\begin{equation*}
\tilde{\mathcal L} = \left(\frac12m\ps  g_{AB} \ps \dot q^A \dot q^B \dot T^{-2} - V({q})\right)\! \dot T  \ps dt.
\label{pld}
\end{equation*}
It is obvious that $\tilde{\mathcal L}$ is time reparametrization invariant, as guaranteed by Theorem~\ref{cov}.

Next we verify (\emph{i}) of Theorem \ref{mainresult} for first order mechanics  in general. The Euler--Lagrange equations for $\tilde L$  with respect to $q^A(t)$ are
\begin{eqnarray*}
0=\frac{\partial \tilde{L}}{\partial q^A} -\frac{d}{dt}\!\left(\frac{\partial \tilde{L}}{\partial \dot q^A}\right) & = & \frac{\partial L}{\partial q^A}\dot{T} -\frac{d}{dt}\!\left(\frac{\partial L}{\partial \dot q^A}\right).
\end{eqnarray*}
Since $T$ is a diffeomorphism, we may use it as a parameter on $\bar {\mathbb R}$; then this becomes
$$
0= \dot{T}\left(\frac{\partial L}{\partial q^A} -\frac{d}{dT}\left(\frac{\partial L}{\partial (d q^A/dT)}\right)\right).
$$
Since $\dot{T}\neq0$ it follows that the Euler--Lagrange equations of $\tilde{L}$ are a time reparametrization of those of $L$ for $q^A(T)$.

As for the Euler--Lagrange equations of $\tilde{L}$ with respect to $T$, we find that \eqref{t} reduces to 
\begin{equation*}
 \frac{\partial L}{\partial T} = \frac{dE}{dT} ,
\end{equation*}
where  $E = -\mathfrak t^0{}_0$ is the energy of the original system. When the mechanical system is scleronomic, this gives conservation of energy.

\subsubsection{Example: The Klein--Gordon Equation, Version I} \label{KGvI}

We first covariantize the Klein--Gordon theory using the Lagrangian density \eqref{KGL1}. Write the covariance field $\eta: \mathbb R^2 \to \mathbb R^2$ as $\eta(t,x) = (T(t,x),X(t,x))$. Since $\phi$ is a scalar field, the action of ${\rm Diff}(\mathbb R^2)$ on $Y = \mathbb R \times \mathbb R^2$ is trivial. Then according to \eqref{newL} and \eqref{cr}, in passing from $L$ to $\tilde L$ we must replace $\phi_{,t}$ and $\phi_{,x}$ by 
$$\phi_{,t} \eta^t{}_T + \phi_{,x} \eta^x{}_T = \frac{\dot \phi X' - \phi' \dot X}{\dot T X' - T' \dot X}$$
and
$$\phi_{,t} \eta^t{}_X+ \phi_{,x} \eta^x{}_X = \frac{-\dot \phi T' + \phi' \dot T}{\dot T X' - T' \dot X}$$
respectively, where we have availed ourselves of Cramer's Rule and have denoted derivatives w.r.t $t$ by dots and $x$ by primes. (Observe in this regard that the Jacobian $J = \dot T X' - T' \dot X$.) Thus the covariantized La\-grangi\-an  is
\begin{equation*}
\tilde L =  \frac{ \phi'\dot\phi (\dot X T' + \dot T X') - \phi'^2 \dot T \dot X-\dot\phi^2  T'X'}{\dot T X' - T'\dot X} - \frac{1}{2}m^2\phi^2(\dot T X' - T'\dot X).
\end{equation*}
Of course, when $\eta = {\rm id}_{\mathbb R^2}$ this reduces to \eqref{KGL1}. A calculation shows that $\tilde L$ does indeed have the expected transformation properties under ${\rm Diff}(\mathbb R^2)$. 

We will give an alternate derivation of this (rather complicated-looking) Lagrangian in \S \ref{KGvII}.

\subsection{The Case When Absolute Objects Are Present}

Suppose the field variables can be split into two types: $\phi = (\psi \ps ;\chi)$, where the $\psi$ are variational and the $\chi$ are not.\footnote{\ To emphasize this distinction we use a semicolon to separate the variational  from the non-variational fields. } (Think of the latter as being background metrics, or tetrads, or Yang--Mills fields, etc., on $X$.) We correspondingly split $Y$ as $V \times_X B$, where the factor $V$ refers to the variational fields and $B$ to the background ones.

We further suppose that the following two Ans\"atze regarding the background fields hold. The $\chi$: (A1) appear only to zeroth order in $\mathcal L$, and (A2) have differential index at most 1.  These conditions are typically satisfied in practice, and we do not expect violations thereof to cause problems, see for instance \cite{CLGoMa2008}. 

Now assume that the given Lagrangian density $\mathcal L$ is diffeomorphism covariant; our problem is to covariantize the theory by eliminating the absolute objects $\chi $ from the ranks of fields.

We accomplish this in a manner that is similar to the technique of \S 3A: We introduce a copy $\bar X$ of the spacetime $X$ into the fiber of the variational part $V$ of the configuration bundle $Y$ and adjoin diffeomorphisms $\eta: X \to \bar X$ as covariance  fields. Next, we view the $\chi$, which originally were sections of $B \to X$,  as sections of  $B \to \bar X$. Thus these quantities are no longer regarded as fields---and hence are no longer absolute objects in the theory---but considered merely as geometric objects $\bar \chi$ anchored to the fiber of the modified configuration bundle $\tilde V = (X \times \bar X) \times_X V$.  Finally, we modify $\mathcal L$ to get the new Lagrangian density $\tilde{\mathcal L}$ on $J^r\tilde V$ defined by
\begin{equation}
\label{newL2}
\tilde{\mathcal L}\ns\left(j^r(\eta, \psi)\right) = \mathcal L\big(j^r\psi\ps ;\eta^{-1}\cdot \bar\chi\big)
\end{equation}
where we have taken Ansatz (A1) into account. Subject to Ansatz (A2), $\tilde {\mathcal L}$ will then remain $r^{\rm th}$-order.

The first  key observation is that the modified theory is also spacetime diffeomorphism covariant:
\begin{thm}
\label{equi2}
The Lagrangian density $\tilde{\mathcal{L}}$ is $\operatorname{Diff} (X)$-covariant.
\end{thm}
\begin{proof}
This is a direct consequence of  definition \eqref{newL2}  and the covariance assumption on $\mathcal L$. Indeed, for $\sigma \in {\rm Diff}(X)$,
\begin{eqnarray*}
\tilde{\mathcal L}\big( j^r\sigma_{\tilde V} \cdot j^r(\eta, \psi)\big) 
& = & \tilde{\mathcal L}\ns \left( j^1\ns (\sigma \cdot \eta),j^r\ns (\sigma \cdot\psi)\right) \\[1 ex]
&= &
{\mathcal{L}} \big( j^r\ns( \sigma \cdot \psi)\, ; (\sigma \cdot \eta)^{-1} \cdot \bar\chi\big) \\[1.5ex]
&= & {\mathcal{L}}\ns \left( j^r\ns( \sigma \cdot \psi)\, ; \sigma \cdot (\eta^{-1} \cdot \bar \chi)\right) \\[1.5ex]
&= & \sigma_* \Big(\mathcal{L}\big(j^r \ns \psi\, ; \eta^{-1} \cdot \bar\chi\big)\Big)\\[1.5ex]
&= & \sigma_*\!\left(\tilde {\mathcal{L}} \big(j^r(\eta, \psi) \big)\right)
\end{eqnarray*} 
where we use \eqref{action on psi} and \eqref{action on eta} to compute
$(\sigma \cdot \eta)^{-1} \cdot \bar \chi = (\eta \circ \sigma ^{-1})^{-1} \cdot \bar{\chi} = (\sigma \circ \eta^{-1})\cdot \bar \chi = \sigma \cdot (\eta^{-1} \cdot  \bar\chi).$
\end{proof}

Thus the modified theory is generally covariant. It remains only to verify that the solution spaces for the original and modified theories are essentially the same. The analogue of Theorem \ref{mainresult} in the present context is

\begin{thm}
\label{mainresult2}
Let $(\eta , \psi)$ be a section of $\tilde{V}$ and let $\bar \chi: \bar X \to B$ be given.
\begin{itemize}
\item[{\rm (\emph{i})}] $(\eta , \psi)$ satisfies the  Euler--Lagrange equations of $\tilde{\mathcal{L}}$ with respect to the fields $\psi$ if and only if $(\psi \, ; \eta^{-1} \cdot \bar  \chi)$ satisfies the  Euler--Lagrange equations of $\mathcal{L}$ with respect to $\psi$, that is,
\[
\frac{\delta \tilde {\mathcal L}}{\delta \psi}(\eta , \psi)=0 \iff \frac{\delta {\mathcal L}}{\delta \psi}(\psi \, ; \eta^{-1} \cdot \bar  \chi)=0.
\]
\item[{\rm (\emph{ii})}] If $(\eta , \psi)$ satisfies the  Euler--Lagrange equations of $\tilde{\mathcal{L}}$ with respect to the fields $\psi$ then the  Euler--Lagrange equations of $\tilde{\mathcal{L}}$ with respect to the covariance fields $\eta$ are vacuously satisfied for $(\eta , \psi)$, that is,
\[
\frac{\delta \tilde {\mathcal L}}{\delta \psi}(\eta , \psi)=0 \Longrightarrow \frac{\delta {\tilde {\mathcal L}}}{\delta \eta}(\eta , \psi)=0.
\]
\end{itemize}
\end{thm}

\begin{proof}
Part (\emph{i}) follows immediately from the definition of $\tilde{\mathcal L}$.

For (\emph{ii}) we consider a compactly supported variation $\eta_\epsilon = F_\epsilon \circ \eta$ of $\eta$, where $F_\epsilon$ is a vertical flow on $X \times \bar X \to X$. We compute from \eqref{newL2}, ${\rm Diff}(X)$-equivariance, and the change of variables formula that
\begin{eqnarray}
\left. \frac{d}{d\epsilon}\right|_{\epsilon =0} \int _X \tilde{\mathcal{L}}\big(j^r (\eta_\epsilon , \psi )\big) 
& = & 
\int _X  \left. \frac{d}{d\epsilon}\right|_{\epsilon =0}  \mathcal{L}\ns\left(j^r\psi\, ; \eta_{\epsilon}{}^{-1}  \cdot \bar \chi\right) \nonumber \\[12pt]
& = & 
\int _X  \left. \frac{d}{d\epsilon}\right|_{\epsilon =0}  \mathcal{L}\ns\left(j^r\psi \, ; \eta^{-1}\cdot ( F_{-\epsilon}  \cdot \bar \chi)\right) \nonumber \\[12pt]
& = & 
\int _X  \left. \frac{d}{d\epsilon}\right|_{\epsilon =0}  \eta^*\Big[\mathcal{L}\ns\left(j^r(\eta\cdot \psi)\, ;  F_{-\epsilon}  \cdot \bar \chi\right)\Big] \nonumber \\[12pt]
& = & 
\int _{\bar X}  \left. \frac{d}{d\epsilon}\right|_{\epsilon =0} \mathcal{L}\ns\left(j^r(\eta\cdot \psi)\, ;  F_{-\epsilon}  \cdot \bar \chi\right) \nonumber \\[12pt]
& = &  \int _{\bar X} \frac{\partial L}{\partial \bar \chi^s} \frac{d(F_{-\epsilon} \cdot \bar \chi)^s}{d\epsilon} \ps d^{\ps n+1}\ns \bar x\label{zap}
\end{eqnarray}
where the partial derivative here is evaluated along the section $\left(j^r(\eta\cdot \psi)\, ;  \bar \chi\right)$. 

Write the infinitesimal generator of the induced flow $\left(F_{\epsilon}\right)_B$ on $B$ as
$$\xi_B= \xi^b\ps \frac {\partial }{\partial x^b} + \xi^s \ps \frac{\partial }{\partial \bar \chi^s}.$$
Then from \eqref{action on psi} and the expansions
$$(\bar \chi \circ F_\epsilon)^s \approx  \bar \chi^s + \epsilon \bar \chi^s{}_{,b}\ps \xi^b \qquad {\rm and} \qquad ((F_{-\epsilon})_B\circ \bar \chi)^s \approx \bar \chi^s - \epsilon \xi^s$$
to order $\epsilon$, we have that
$$\frac{d(F_{-\epsilon} \cdot \bar \chi)^s}{d\epsilon} = \bar \chi^s{}_{,b}\ps \xi^b - \xi^s.$$
(This expression is the vertical part of the vector field $\xi _B$ with respect to the splitting $T_{\bar \chi(\bar x)} B = V_{\bar \chi(\bar x)} B \oplus {\rm im}(T_{\bar x}\bar \chi)$ induced by $j^1 _{\bar x} \bar \chi$.)
As $\bar \chi$ has differential index $\leq 1$, we may write
$$\xi^s = C^s{}_b\ps \xi^b + C^{sa}{}_b\ps \xi^b{}_{,a}$$
for some functions $C$ on $B$ (see \cite{GoMa1992} for specifics). Plugging into \eqref{zap}, we end up with
\begin{eqnarray*}
\left. \frac{d}{d\epsilon}\right|_{\epsilon =0} \int _{X} \tilde{\mathcal{L}}\ns \left(j^r (\eta_\epsilon , \psi )\right) 
& = & \int_{\bar X}  \frac{\partial L}{\partial \bar \chi^s}  \Big( \left[ \bar \chi^s{}_{,b} - C^s{}_b\right] \xi^b - C^{sa}{}_b\ps \xi^b{}_{,a}  \Big)d^{\ps n+1}\ns \bar x\nonumber\\[6pt]
& = & \int_{\bar X}  \frac{\delta L}{\delta \bar \chi^s}  \Big( \left[ \bar \chi^s{}_{,b} - C^s{}_b\right] \xi^b - C^{sa}{}_b\xi^b{}_{,a}  \Big)d^{\ps n+1}\ns \bar x
\end{eqnarray*}
where we have used (A1) in going from the first line to the second.

Now suppose the $\psi$ are on shell, so that by part (\emph{i}) the pair $(\psi \, ;\eta^{-1}\cdot \bar \chi)$ satisfies the Euler--Lagrange equations of $\mathcal L$ with respect to $\psi$. We observe that since $\mathcal L$ is ${\rm Diff}(X)$-covariant,  the pair $(\eta \cdot \psi\, ; \bar \chi)$ also satisfies the Euler--Lagrange equations of $\mathcal L$ with respect to $\psi$ (i.e., $\eta$ is a symmetry). Bearing this in mind, let $\mathfrak T$ be the SEM tensor density of $\mathcal L$; then the generalized Hilbert formula (3.3) in \cite{GoMa1992} and Prop. 1 \emph{ibid} give 
\begin{equation*}
\mathfrak T^a{}_b = -C^{sa}{}_b \frac{\delta L}{\delta \bar \chi^s} \qquad {\rm and} \qquad 
D_a \mathfrak T^a{}_b = \left( \bar \chi^s{}_{,b} - C^s_{b}\right) \frac{\delta L}{\delta \bar \chi^s}
\end{equation*}
along the section $\left(j^r(\eta\cdot \psi)\, ;  \bar \chi\right)$, respectively. Thus the r.h.s.  of the expression above reduces to
\begin{eqnarray*}\int_{\bar X} \Big( (D_a\mathfrak T^a{}_b)\xi^b + \mathfrak T^a{}_b\ps \xi^b{}_{,a}\Big)\ps d^{\ps n+1}\ns \bar x
& = & \int_{\bar X} D_a\! \left(\mathfrak T^a{}_b \ps \xi^b\right)\ps d^{\ps n+1}\ns \bar x \\[8pt]
& = & \int_{X} D_a\! \left(\mathfrak T^a{}_b\ps \xi^b\right)(\det J)\, d^{\ps n+1}\ns x
\end{eqnarray*}
where we have again used the change of variables formula. 

But now a computation along with the Piola identity yields
$$D_a(\mathfrak T^a{}_b\ps \xi^b)(\det J) = D_\mu(\mathfrak P^\mu{}_b \ps \xi^b)$$
where $\mathfrak P^\mu{}_b = x^\mu{}_a \mathfrak T^a{}_b (\det J)$ is the Piola-Kirchhoff SEM tensor density of $\mathcal L$. Thus
$$\left. \frac{d}{d\epsilon}\right|_{\epsilon =0} \int _X \tilde{\mathcal{L}}\ns \left(j^r (\eta_\epsilon , \psi )\right) 
= \int_X  D_\mu(\mathfrak P^\mu{}_b\ps \xi^b)\ps d^{\ps n+1}x$$
which vanishes by virtue of the divergence theorem, since by assumption $\xi^b(x,\bar x)$ is compactly supported as a function of $x \in X$.
\end{proof}

Thus for given $\bar \chi$ the map $(\eta,\psi) \mapsto \big(\eta,(\psi \, ; \eta^{-1}\cdot \bar \chi)\big)$ provides an isomorphism of $\tilde{\mathcal S}$ with ${\rm Diff}(X) \times \mathcal S$.


\subsubsection{Example: The Klein--Gordon Equation, Version II} \label{KGvII}


We covariantize the second version of the Klein--Gordon Lagrangian density \eqref{KGL2}. View the metric as being anchored to the fiber $\bar X$, so that $g^{\mu\nu} =  x^\mu{}_a x^\nu{}_b \ps \bar g^{ab}$, and take
$$\tilde{L} = \left(x^\mu{}_a x^\nu{}_b \ps \bar g^{ab}  \phi_{,\mu}\phi_{,\nu} - \frac{1}{2}m^2\phi^2\right)(\det J).$$

Choosing coordinates on $\bar X$ in which 
$$\bar g = 
\left(
\begin{array}{ccc}
 0 &   1   \\
  1 &   0
\end{array}
\right)
$$
and again writing $\eta = (T,X)$, one finds that $\tilde L$ reduces to the Lagrangian of \S \ref{KGvI}.

The case when a metric on spacetime is the only background field was exhaustively analyzed in \cite{CLGoMa2008}. There, it was shown that the vanishing of the Euler--Lagrange expressions for the covariance fields could be viewed as a consequence of the Bianchi identity $\nabla_{\!\mu}\mathfrak T^\mu{}_\nu = 0$ for the Hilbert SEM tensor density $\mathfrak T^{\mu\nu} = -2\, \delta \mathcal L /\delta g_{\mu\nu}$.


\section{Enlarging Vertical Automorphism Groups}

Now we consider the case when a group $\mathcal{G}$ acts \emph{vertically} on the configuration bundle $Y$, that is, the elements of $\mathcal{G}$ are $\pi_{XY}$-bundle automorphisms $Y \to Y$ covering the identity on $X$. It can happen that $\mathcal G$ is not an invariance group of  the Lagrangian density $\mathcal L$ even though, for instance,  (\emph{i}) $\mathcal G$ acts by symmetries,  or (\emph{ii}) some subgroup  is an invariance group. Thus it may be desirable to find another essentially equivalent field theory for which  $\mathcal G$ is an invariance group.


\subsection{The Construction}\label{secf}

As with diffeomorphism covariance, we will define a new bundle $\tilde{Y}\to X$ and a new Lagrangian density $\tilde{\mathcal{L}}$ admitting $\mathcal{G}$ as an invariance group. For this we make the additional assumption that $\mathcal{G}$ can be identified with the set of sections of certain bundle $E \to X$, so that if $\eta \in \mathcal G$ then $\eta(x) \in E_x \subset {\rm Diff}(Y_x)$. (This is the case, in particular, in gauge theory; see \S \ref{utiyama}. In other situations, one may wish to identify $\mathcal G$ with some \emph{subset} of sections of $E$.) Mimicking the constructions in the previous section we achieve this by setting\footnote{ \  Analogous to \S3, we here take $k$ to be the order of the highest derivative of $\eta  \in \mathcal G$ that appears in $\eta_Y\circ\phi$, where  $\eta$ is considered a a section $X\to E$. }
\[
\tilde{Y}=E\times _X Y
\]
and
\begin{equation}
\label{vertL}
\tilde{\mathcal L}\big(j^{k+r}(\eta,\phi)(x)\big)={\mathcal L}\big(j^r(\eta_Y \circ \phi)(x)\big).
\end{equation}
Thus sections $\eta$ of $E$ will play the role of covariance fields in the present context. 
 
 Now the group $\mathcal{G}$ acts on the new bundle in a natural way by
\begin{equation}
\sigma _{\tilde{Y}}\big((\eta(x),\phi(x)) \big) = \big(\eta(x) \circ \sigma(x)^{-1},\sigma _Y(\phi(x)) \big)
\label{vertaction}
\end{equation}
for any $(\eta(x),\phi(x))\in \tilde{Y}$ and section $\sigma$ of $E\to X$ (compare \eqref{action on eta}). With respect to this action we immediately verify, in the manner of Theorem \ref{cov}, that the modified Lagrangian density is invariant under this $\mathcal{G}$-action:  for all $\sigma \in \mathcal{G}$ and $\gamma \in J^{k+r} \tilde Y$, 
\begin{equation*}
\mathcal{\tilde{L}}\ns \left( j^{k+r}\ns \sigma_{\tilde Y}(\gamma) \right) = 
\mathcal{\tilde{L}}(\gamma).
\end{equation*}
Note that, in contrast to \eqref{equi}, there is no push-forward on the r.h.s. as the group action is vertical. (Thus we speak of \emph{in}variance in this context, rather than \emph{co}variance.)



Following the now familiar pattern, we find that the Euler--Lagrange 
equations of $\tilde{\mathcal{L}}$ with respect to the fields $\phi$ are essentially the same as those of the original 
Lagrangian $\mathcal{L}$. With respect to the new variables $\eta$, no more conditions are added, that is, they
do not contribute to the dynamics of the variational problem. In fact, Theorem \ref{mainresult} holds word-for-word in this context, and the proof is also similar, but simpler, as $\eta_X = {\rm id}_X$.

\subsection{Examples}

We provide several illustrations of our technique, beginning with an example of circumstance (\emph{i}) above.


\subsubsection{Abelian Chern--Simons Theory}

Treating the connection $A$ as an ordinary 1-form on a 3-manifold $X$, the Lagrangian density is
\begin{equation*}
\mathcal L(j^1\ns A) = dA \wedge A.
\end{equation*}

The additive group $C^\infty(X)$, thought of as sections $f$  of $X \times \mathbb R \to X$, acts on $T^*X$  according to
\begin{equation}
f_{T^*\ns X}(A_x)= A_x + df(x).
\label{cs}
\end{equation}
for $A_x \in T_x ^* X$. 
The Lagrangian density is not  invariant under this action; nonetheless, \eqref{cs} is a symmetry: if $A$ is a solution (i.e., $dA = 0$), any $f _{T^*X}\cdot A=
A + df$ will be  as well. So it is natural to search for an equivalent field theory which is $C^\infty(X)$-invariant.

Our constructions produce the modified configuration bundle 
$$\tilde Y = (X \times \mathbb R) \times_X T^*X$$ 
and the (first order) Lagrangian density
$$\tilde {\mathcal L}\big(j^1\ns (\eta,A)\big) = dA \wedge (A+d\eta).
$$
As expected, this is invariant under the action \eqref{vertaction} which here takes the form
\begin{equation}
\label{emgauge}
f_{\tilde Y}(\eta,A)= (\eta - f,A + df).
\end{equation}
The Euler--Lagrange equation for $A$ is now
$d(2A + d\eta)=0  \iff dA=0$ and that for $\eta$ is vacuously satisfied,
consistent with our results and the fact that $C^\infty(X)$ is a symmetry group of the original system. 

\begin{rmk} \rm \label{nontriv}
All this can be generalized to the case  when $A$ is regarded as a connection on a trivial principal 
 $G$-bundle over $X$, with $G$ not necessarily Abelian. The Lagrangian density becomes
\[
L= {\rm tr}\ns \left(\ns A\wedge\Omega^{A}-\frac{1}{6}A\wedge\lbrack A,A]\ns \right)
\]
where $\Omega^{A}=dA+[A,A]$ is the curvature.  For non-trivial bundles, a variational theory can be given in terms of the bundle of connections as the configuration bundle $Y$ and the
adjoint bundle ${\rm Ad}(G)$ as the bundle $E$ whose sections are symmetries. See \S\ref{utiyama} following, and \cite{CM} and \cite{T} for more details. For higher dimensional manifolds $X$, the situation is more complex. In any case, good references are \cite{F1}
and \cite{F2}. 
 \hfill $\blacklozenge$  
\end{rmk}
 
\subsubsection{The St\"uckelberg Trick}

We consider the Proca field $A$, a 1-form on 4-dimensional spacetime $(X,g)$, with Lagrangian density $\mathcal{L}:J^1(T^*X)\to \wedge ^4 X$ being
\begin{equation}
\label{procaL}
\mathcal{L}(j^1\! A) = \frac12\big(-\|dA\|^2 + m^2\|A\|^2\big) \ps \sqrt {- \det g} \, d^{\ps 4}x.
\end{equation}
Here the norms are taken with respect to the metric $g$. The Proca equations are $\star \: d\star dA = m^2 A$ where $\star$ denotes the Hodge star operator of $g$.

Thinking of the Proca field as being electromagnetism but with a massive photon, we see that the mass term in $\mathcal{L}$ breaks the electromagnetic gauge invariance $A \mapsto A + df$. We may restore this invariance via our construction, even though these electromagnetic shifts are not even symmetries of the Proca equations.

Formally, everything works as in the Chern--Simons example. We have
$$\tilde {\mathcal L}\big(j^1\ns (\eta,A)\big) = \frac12\Big(-\|dA\|^2 + m^2\|A+ d\eta\|^2\Big) \ps \sqrt {- \det g} \, d^{\ps 4}x $$
which is obviously invariant under \eqref{emgauge}.  This is the ``St\"uckelberg Lagrangian'' and the passage from the Proca Lagrangian to it is known as the ``St\"uckelberg trick'' \cite{Stuckelberg1957}. The Euler--Lagrange equations for $A$ are now 
$$\star \: d\star dA = m^2 (A+d\eta)$$ 
and, using these, we find that the field equation $d\star(A+d\eta) = 0$ for the ``St\"uckelberg scalar'' $\eta$ is identically satisfied, as expected.  

Although the St\"uckelberg trick seems innocuous enough, the corresponding field theories are structurally quite different: the St\"uckelberg theory is purely first class (in the sense of Dirac; see \cite{GoMa2010}) while the Proca theory is purely second class. (Indeed, the historical importance of the St\"uckelberg trick was to effect this change, cf. \cite{BBG1995}.) Despite this, the two theories are essentially equivalent as defined above.

\subsubsection{The Minimal Coupling (or Utiyama) Construction}\label{utiyama}

\noindent The main instance of circumstance (\emph{ii}) above is gauge theory. In this context, we have  a field theory invariant under the action of a finite-dimensional Lie group $G$ which we wish to `gauge,' as we now explain.

 Let $P \to X$ be a principal $G$-bundle. \emph{Gauge transformations} are vertical diffeomorphisms $\Phi : P \to P$ such that $R_g \circ \Phi = \Phi \circ R_g$ for any $g \in G$, where $R_g$ stands for the right action of $G$ on $P$. These transformations can be seen as sections of the \emph{Adjoint bundle} ${\rm Ad}({G}) \to X$; this is the bundle associated to $P$ with respect to the conjugate action of $G$ onto itself. More precisely, ${\rm Ad}({G})=(P\times G)/G$ where the action is $(p,h)\cdot g = \big(R_g (p) , g^{-1}hg\big)$, for any $(p,h)\in P \times G$, $g \in G$. 
Let $(p,g)_G$ stand for the class of $(p,g)$ in ${\rm Ad}({G})$. Given a section $\eta$ of ${\rm Ad}({G})$, the mapping $\Phi_\eta : P \to P$ defined as
\[
\Phi_\eta (p) = R_g (p), \quad \text{where} \quad \eta(\pi _{XP}(p))=(p,g)_G
\]
is readily seen to be equivariant and vertical, that is, a gauge transformation. Moreover, any gauge transformation can be identified with a section of the Adjoint bundle via the equation above. We denote the collection of all sections of ${\rm Ad}({G}) \to X$ by $\mathcal G$; it is the \emph{gauge group}.

Let $V$ be a vector space on which $G$ acts linearly and $Y = (P \times V)/G$ be the associated vector bundle. Then $G$ acts $Y$ on the left by vertical bundle automorphisms according to 
$$g\cdot (p,v)_G=(p,g\cdot v)_G = (R_{g^{-1}} p,v)_G.$$
Consequently, $\mathcal G$ also acts on $Y$ according to 
$\eta\cdot (p,v)_G=(\Phi_\eta(p),v)_G$. Thus we have \emph{gauged} the action of $G$ on $Y$: we have extended the `global' (i.e., base-independent) action of the (finite-dimensional) group $G$ by that of the `local' (i.e., base-dependent) action of the (generally infinite-dimensional) group $\mathcal G$.


The situation encountered in gauge theory is when the associated  vector bundle $Y \to X$ is the configuration bundle of some field theory with Lagrangian density $\mathcal L$. Typically $\mathcal L$ is $G$-invariant, but \emph{not} $\mathcal G$-invariant. However, one requires invariance under $\mathcal G$ for physical reasons; we now apply our prescription to achieve this. Thus we take 
$$\tilde Y = {\rm Ad}(G) \times_X Y$$
and define the modified Lagrangian density $\tilde{\mathcal L}$ by \eqref{vertL}.

We refer to a local trivialization $U \times G$ of ${\rm Ad}(G) \to X$ to obtain  a clearer understanding of  this construction when $r =1, \ k=0$. Using matrix notation and the Leibniz rule, from \eqref{vertL} we obtain the (first order) Lagrangian 
\begin{eqnarray*}
{\tilde{L}}\big(x^\mu ,  \eta ^A{} _B , \eta^A{} _{B\mu}, y ^A , y ^A {}_ \mu \big)  & = & 
{L}\big(x^\mu , \eta^A{} _B y ^B , \eta^A{} _B y ^B{} _\mu + \eta^A{} _{B\mu} y ^B\big)\\[4pt]
& = &{L}\big(x^\mu , \eta^A{} _B y ^B , \eta^A{} _B \big[y ^B{} _\mu + (\eta^{-1})^B{}_D\eta^D{} _{C\mu} y ^C\big]\big)
\end{eqnarray*}
where we have isolated a factor of $\eta^A{}_B$ in the last argument. Now by assumption, ${L}$ is $G$-invariant:
\begin{equation*}
{L}\big(x^\mu , \eta ^A{} _B y ^ B , \eta ^A{} _B y ^B {}_\mu\big) ={L}\big(x^\mu , y ^A , y ^A{} _\mu\big).
\end{equation*}
Applying this to the r.h.s. of the preceding equation gives
\begin{eqnarray}
\label{GTL}
{\tilde{L}}\big(x^\mu , \eta ^A{} _B , \eta^A{} _{B \mu }, y ^A , y ^A{} _ \mu \big) & = & 
{L}\big(x^\mu , y ^A , y ^A{} _\mu + (\eta^{-1})^A {}_D \eta^D{} _{C\mu} y ^C\big).
\end{eqnarray}
Note that  $(\eta^{-1})^A{} _{D } \eta^D{} _{C\mu} $ are the components of the matrix $\eta^{-1}\ps d\eta$, and as $d\eta$ is pointwise an element of the bundle $T^* U \otimes TG$, this expression is nothing but the identification of $d\eta$ as belonging to $T^* U \otimes \mathfrak{g}$ corresponding to the identification $TG = G \times \mathfrak g$ given by $(g,\dot g) \mapsto (g,TL_{g^{-1}}\cdot \dot g)$, where $\mathfrak g$ is the Lie algebra of $G$. If we write $A$ for the variables of $T^*U \otimes \mathfrak{g}$, we have finally
\[
{\tilde{L}}\big(x^\mu , \eta ^A{} _B , \eta^A{} _{B \mu }, y ^A , y ^A{} _ \mu \big) = 
{L}\big(x^\mu , y ^A , y ^A {}_\mu + A_{\mu B}^A y ^B \big).\]

On the other hand, the first jet bundle of $U \times G \to U$ is $T^*U \times TG \to U \times G \to U$.
Thus with the identification above we can write
\[
J^1(U \times G) = G \times (T^* U  \otimes \mathfrak{g}).
\]
At this point we observe  that even though the modified Lagrangian density $\mathcal{\tilde{L}}$ was initially defined (with respect to the given local trivialization) on the jet space
\[J^1 \tilde{Y}|\ps U = J^1 Y|\ps U \times (T^*U \otimes TG) = J^1 Y|\ps U \times G \times (T^*U \otimes \mathfrak{g}),
\]
it has no explicit dependence on $\eta$. We can thus drop the  factor of $G$ in the fiber on which $\mathcal{\tilde{L}}$ is defined and simply view
\begin{equation}
\label{jui}
\mathcal{\tilde{L}} : J^1 Y|U \times (T^*U \otimes \mathfrak{g}) \to \wedge ^{n+1} U.
\end{equation}

We now notice that:
\begin{enumerate}
\item The bundle $T^* U \otimes \mathfrak{g}$ is the bundle of connections of the principal bundle 
$P\ps | \ps U=U \times G$ of which $Y |\ps U = U \times V$ is an associated vector bundle.

\item The expression $\nabla_{\ns \mu}\phi^A = \phi ^A {}_{,\mu} + (\eta^{-1})^A {}_D \eta^D{} _{C,\mu} \phi ^C$ is the covariant derivative of the field $\phi$ 
with respect to the connection $A= \eta^{-1} \ps d\eta$.

\item From \eqref{GTL} the variational equations $\delta \tilde{L}/\delta \phi^B\ps (\eta,\phi) = 0 $ can be written 
\begin{equation}
\label{delEL}
\frac{\partial L}{\partial \phi^B} - {\nabla}_{\ns \mu} \!\left(\frac{\partial L}{\nabla_{\ns \mu} \phi^B} \right)= 0.
\end{equation}
\end{enumerate}
This is the well-known \emph{minimal coupling} (or \emph{Utiyama}) trick. Our prescription thus transforms the $G$-invariant Lagrangian density $\mathcal{L}$ into a $\mathcal G$-invariant Lagrangian density $\tilde{\mathcal{L}}$ by changing partial derivatives into covariant derivatives with respect to a connection. This connection $A$ (and \emph{not} the field $\eta$) is the new field of the modified Lagrangian according to the classical gauge-theoretic prescription, while \eqref{delEL} are recognized as the Euler--Lagrange equations (for the matter fields) in gauge theory.

Some remarks are in order:
\begin{rmk} \rm
We have seen in \S \ref{secf} that the Euler--Lagrange equations of  $\mathcal{\tilde{L}}$ for the 
fields $(\eta, \phi)$ are equivalent to the Euler--Lagrange equations of $\mathcal{L}$ for $\eta\cdot \phi$. This is of course still true. But if instead we consider the fields on which the modified Lagrangian density depends to be $A$ and $\phi$, this result no longer holds. Indeed, the variational equation with respect to $A$ forces
\[
0 = \frac{\partial{\tilde{L}}}{\partial A_\mu} = \frac{\partial L}{\partial \phi^A{}_{,\mu}}\ps\phi^A
\]
which is in general inconsistent with the dynamics of the original variational problem (compare \eqref{EMSEM}). This is why the minimal coupling trick requires the addition of a source term depending on the  connection $A$ and its first derivatives to the Lagrangian density. This term must be gauge-invariant; while the Yang--Mills Lagrangian density is the usual choice, a result of \cite{Utiyama1956} (see also \cite{GP1977}) provides a characterization of all such Lagrangians in term of the curvature  of the connection.

On the other hand, it is amusing to wonder what would happen if we treated $A$ as simply a background field, and then generally covariantized the theory as in \S 3B? `Second covariantization'?
 \hfill $\blacklozenge$  
\end{rmk}

\begin{rmk} \rm
Connections $A$ of the type $\eta^{-1}\ps d\eta$ are flat. Thus the prescription we provide here is 
the minimal coupling trick where the new field (the connection) has no curvature. As the curvature is understood as the field strength of $A$, this explains physically why for these connections, the Euler--Lagrange equations of the covariantized theory will be equivalent to those of the original theory:  a gauge transformation can be used to reduce $A=0$.
 \hfill $\blacklozenge$  
\end{rmk}

\begin{rmk} \rm
 The condition for $A$ (as a section of $T^* U \otimes \mathfrak{g} \to U$) to be flat is equivalent to the covariance field $\eta$ being a holonomic section of  $J^1 {\rm Ad}(G)|\ps U  = T^* U 
\otimes TG = G \times (T^* U \otimes \mathfrak{g})$. Hence we could say that the classical minimal coupling trick just takes the construction of \S \ref{secf} and assumes that the connection $A$ need not be holonomic. This can be seen, somehow, as natural because the new Lagrangian density $\mathcal{\tilde{L}}$ depends only on $A$, not on $\eta$, and hence we could `forget' that $A$ is  induced by the holonomic section $\eta$.
 \hfill $\blacklozenge$  
\end{rmk}

\section{The General Case}

Finally, suppose the group of $\pi_{XY}$-bundle automorphisms under consideration acts neither `horizontally' (as in \S 2) nor purely vertically on $Y$  (as in \S 3). For instance, consider a field theory based on a trivial principal bundle $X \times G \to X$, which is not generally covariant  
and which we wish to gauge. Since ${\rm Diff}(X)$ acts on $X \times G$ and both it and the gauge group $ C^\infty(X,G)$ can be realized as sections of bundles over $X$, we may straightforwardly extend our technique to this case, regarding the group $ \mathcal G = {\rm Diff}(X) \ltimes C^\infty(X,G)$ as sections of $X \times X \times G \to X$, and then covariantizing with respect to $\mathcal G$.

The case when the relevant group $\mathcal G$  cannot be realized as sections of some bundle over $X$ is more complicated; nonetheless, our constructions still work with the following adjustments. As in  \cite{CLGoMa2008b} we ``concatenate'' $Y\times Y \to Y$ and $Y \to X$ to obtain the modified configuration bundle $Y \times Y \times_X Y \to Y \times_X X = Y$. Sections of this bundle will consist of pairs $(\eta,\tilde \phi)$, where $\eta \in \mathcal G$ is a $\pi_{XY}$-bundle automorphism and the ``fields'' $\tilde \phi : Y \to Y$ are related to the original fields by $\tilde \phi = \phi \circ \pi_{XY}$.  The automorphisms $\eta: Y \to Y$ are the covariance fields in this approach.

As the modified Lagrangian density on $J^{r+k}(Y \times Y\times_X Y)$ we take
\begin{equation*}
\tilde{\mathcal L}\big(j^{r+k}(\eta,\tilde \phi)(y)\big) = \eta_{X}{}^*\big[\mathcal L\big(j^r(\eta\cdot \phi)(\eta_X(x))\big)\big]\ps \delta^N\!\big(y-\phi(x)\big)
\end{equation*}
where $x = \pi_{XY}(y)$ and $N$ is the fiber dimension of $Y\to X$. As the delta function is a scalar density of weight 1 along the fibers of  $Y\to X$, and $\mathcal L$ is is a scalar density of weight 1 on $X$, $\tilde{\mathcal L}$ is a scalar density of weight 1 on $Y$. These observations coupled with a computation along the lines of that contained in the proof of Theorem \ref{cov} shows that  $\tilde{\mathcal L}$ is $\mathcal G$-covariant.

Finally, Theorem \ref{mainresult} remains valid as stated, the proof needing only minor modifications.


\paragraph{\Large Acknowledgments.}  We would like to thank Jair Koiller, Melvin Leok, Jerry Marsden, Peter Michor, Joris Vankerschaver and Hiroaki Yoshimura for their valuable insights.

\end{document}